\documentclass[lettersize,journal]{IEEEtran}
\IEEEoverridecommandlockouts

\usepackage{amssymb}
\usepackage[cmex10]{amsmath}
\usepackage{stfloats}
\usepackage{graphicx}
\usepackage{subfigure}
\usepackage{tabularx}
\usepackage{epsfig,epsf,color,balance,cite}
\usepackage{verbatim}
\usepackage{url}
\usepackage{bm}
\usepackage{booktabs}
\usepackage{siunitx}
\usepackage{amsmath}
\allowdisplaybreaks[2]
\usepackage{ntheorem}

\newtheorem{lemma}{Lemma}
\newtheorem*{proof}{Proof}

\usepackage{algorithm}
\usepackage{algorithmic}
\usepackage{color}
\definecolor{myc1}{rgb}{0,0,0}

\begin{document}

% % paper title
\title{Joint Communication and Computation Design for Mobile Embodied AI Network (MEAN)}

\author{Chenliang Wu,
        Zhouxiang Zhao,
        Jiaxiang Wang,
        Ruopeng Xu,
        Chen Zhu,\\
        Zhaohui Yang,
        and Zhaoyang Zhang,~\IEEEmembership{Senior Member,~IEEE}

% % \thanks{This work is supported by . \textit{(Corresponding author: Zhaohui Yang)}.}
\thanks{Chenliang Wu, Zhouxiang Zhao, Ruopeng Xu, Zhaohui Yang, and Zhaoyang Zhang are with the College of Information Science and Electronic Engineering, Zhejiang University, Hangzhou 310027, China (e-mails: \{chenliangwu, zhouxiangzhao, ruopengxu, yang\_zhaohui, ning\_ming\}@zju.edu.cn).}
\thanks{Jiaxiang Wang is with the Department of Engineering, King's College London, London, UK. (email: jiaxiang.wang@kcl.ac.uk).}
\thanks{Chen Zhu is with the School of Communication Engineering, Hangzhou Dianzi University, Hangzhou 310018, China, and also with the Polytechnic Institute, Zhejiang University, Hangzhou 310015, China (e-mail: zhuc@zju.edu.cn).}
% % \thanks{W. Xu is with National Mobile Communications Research Laboratory, Southeast University, Nanjing 211189, China (e-mail: wxu@seu.edu.cn).}
\vspace{-1em}
}

% make the title area
\maketitle

\begin{abstract}
This letter investigates the problem of energy efficient collaborative strategy for mobile embodied artificial intelligence network (MEAN) over wireless communication. In the considered model, the agents execute the tasks through collaboration, and they can switch between two operating modes based on the signal-to-noise ratio (SNR) and global collaboration. The dual-mode comprises the base station (BS)-assisted collaborative mode, in which agents make decisions through semantic communication with BS and then collaborate on tasks, and the local computing mode, in which the agents make decisions and execute tasks independently. Due to the dynamic wireless communication and flexible collaboration strategy, we jointly consider computation energy, communication energy, and task-execution energy with specific collaborative gains into a mixed-integer nonlinear programming (MINLP) optimization problem whose goal is to minimize the total system energy consumption. To solve it, we propose a lower-complexity enumeration algorithm: first, we get the optimal closed-form solution for semantic compression ratio and transmit power by proving the strict convexity. Second, we determine the scale of collaboration and the operating mode of each agent by a greedy sorting algorithm based on individual energy-saving potentials. Simulation results show that the proposed algorithm can significantly reduce the total energy consumption compared to benchmark schemes.

\end{abstract}


\IEEEpeerreviewmaketitle

\section{Introduction}
\IEEEPARstart{D}{riven} by rapid advancements in large language models (LLMs) and robotics, the agents equipped with autonomous analytical capabilities have garnered significant research attention \cite{ijcai2024p890}. The mobile embodied artificial intelligence network (MEAN) consists of multiple such agents and is capable of solving various complex problems more efficiently and robustly through collaboration by information exchange and task allocation among the agents \cite{11242028,zhao2026agenticaiempoweredwirelessagent}. Due to the benefits of collaboration, MEAN outperforms single agent in these complex applications including uncrewed aerial vehicle (UAV) swarms for search operations \cite{10254323}, multi-quadruped robot teams for scene modeling \cite{10265226}, multi-manipulator assembly lines \cite{10103893} and so on. However, many existing control frameworks simplify the system design by assuming ideal communication conditions \cite{yang2025multiagent}, which only concentrate on how to effectively allocate task resources to promote collaboration in MEAN \cite{10240696}.

In practical deployments, information exchange among agents is severely bottle-necked by unreliable wireless communication with finite bandwidth, multi-path fading, and dynamic interference \cite{11006980}. Although wireless communication involves uncertainties, the mobility of MEAN requires the use of wireless communication for collaboration, especially in outdoor tasks. Therefore, it is a challenge for energy efficient collaboration in MEAN through wireless communicaiton. On the one hand, the agent exchanged large date size with images and videos, which will result in significant latency without data processing, and outdated data may hinder collaboration. On the other hand, the energy of the agent is limited, so the energy efficiency must be improved to meet the demands of computing, communications, and task execution. To face the challenge, semantic communication is considered one of the key solutions, which can reduce the communication payload and transmission energy by extracting and transmitting only task-relevant semantic features instead of raw sensory data compared to traditional bit-level communication \cite{wang2025generative}. Meanwhile, the appropriate collaborative strategy can also save energy \cite{11022590}.

Therefore, combining semantic communication and flexible collaboration strategy, this letter proposes an energy efficient collaborative strategy for MEAN over wireless communication to maximize energy efficiency. The main contributions of this work are summarized as follows:
\begin{itemize}
    \item we investigates a MEAN over wireless communication which consist of a base station (BS) and a set of agents. The agent can switch between BS-assisted collaborative mode and local computing mode, which based on the signal-to-noise ratio (SNR) and global collaboration.
    \item Taking into account the benefits of collaboration gains, we formulate a joint optimization problem encompassing the semantic compression ratio, transmit power, and operating mode selection to maximize the energy efficiency under processing time constraints, which is a mixed-integer nonlinear programming (MINLP) problem.
    \item To solve the MINLP problem, we propose a lower-complexity enumeration algorithm: optimizes the continuous computation and communication resources, and determines the discrete operating modes through a sorting strategy based on optimal energy-saving potentials.
\end{itemize}

\section{System Model and Problem Formulation}\label{Sec:smpf}
We consider an uplink communication with MEAN, which consists of a BS and a set of agents indexed by $\mathcal{N}=\{1,2,\dots,N\}$. Each agent is capable of sensing, semantic communication, local computation and task execution, while the BS serves as a central node for data aggregation and global collaboration. During one collaboration round in a task, each agent first processes its sensory data within limit processing time \(T_0\), and then executes the assigned task. Our goal is to minimize the energy consumption of the MEAN.

To account for the dynamic wireless communicaiton and flexible collaboration strategy, we adopt a dual-mode operation strategy based on SNR and global collaboration. We use a binary index $x_i \in \{0,1\}$ to denote the operate mode of agent $i$. As shown in Fig.~\ref{fig:SystemModel}, \(x_i=1\) indicates the BS-assisted collaborative mode, in which agent $i$ semantically communicates the sensed data with the BS and executes the task with global collaboration. In contrast, \(x_i=0\) indicates the local computing mode, in which agent $i$ processes the sensed data and executes the task independently.

We assume the the channel $h_i$ between the BS and agent $i$ remains constant during the limit processing time \(T_0\). Then, the maximum achievable SNR of agent $i$ is $\gamma_i=\frac{P_{\max}|h_i|^2}{\sigma^2}$, where $\sigma^2$ is the additional white Gaussian noise and $P_{\mathrm{max}}$ is the maximum transmit power shared by all agents. We define the channel-feasibility indicator as
\begin{equation}
a_i=
\begin{cases}
1, & \gamma_i\ge \gamma_{\mathrm{th}},\\
0, & \gamma_i<\gamma_{\mathrm{th}},
\end{cases}
\label{eq:a_i}
\end{equation}
where $\gamma_{\mathrm{th}}$ is a predefined threshold which determines whether BS-assisted collaborative mode is feasible.

When $a_i=1$, agent $i$ experiences a favorable SNR and is able to operate in either the BS-assisted collaborative mode or the local computing mode, which depends on the global collaboration. Conversely, when $a_i=0$, agent $i$ experiences a unfavorable SNR and is restricted to the local computing mode. The correspondence between $x_i$ and $a_i$ is $x_i \le a_i$.

\begin{figure*}
    \centering
    \includegraphics[width=1\linewidth]{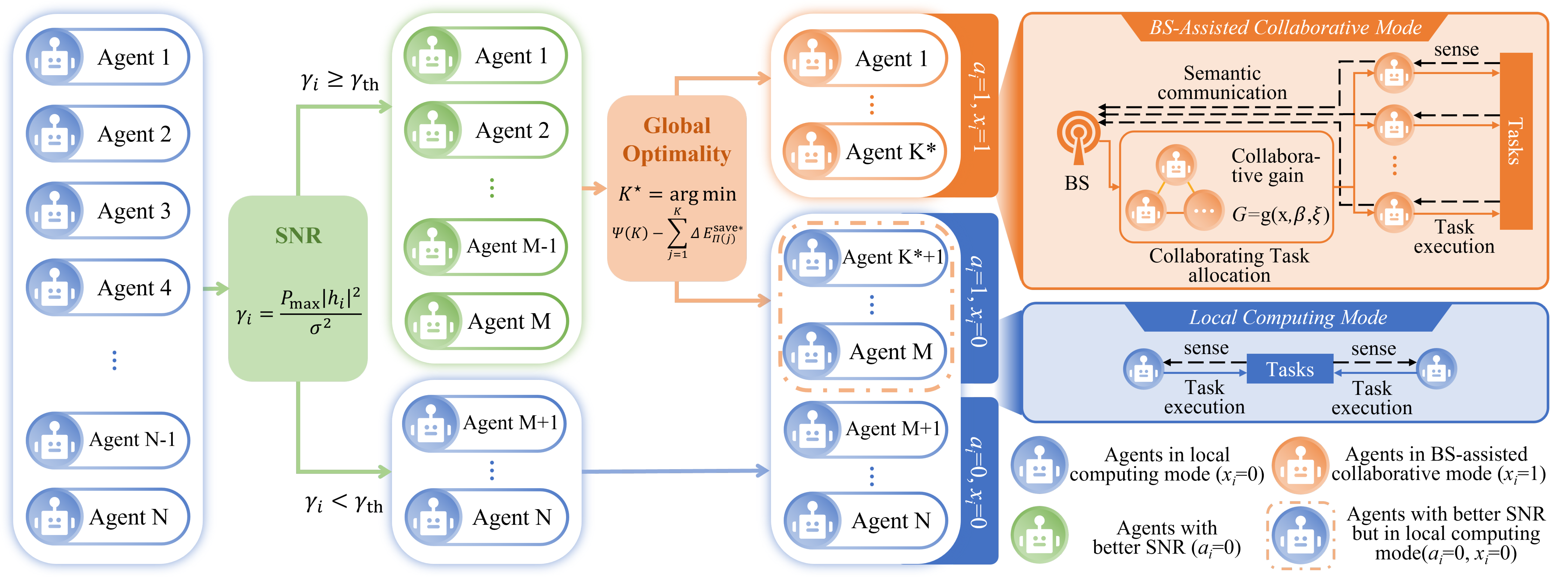}
    \caption{System model and the dual-mode operation strategy in MEAN.}
    \label{fig:SystemModel}
    \vspace{-1em}
\end{figure*}

\subsection{Data Process Model}
In BS-assisted collaborative mode, agent $i$ performs semantic extraction on raw sensory data of size \(D_i\) bits. If we denote the compressed semantic data by $C_i$, the compression ratio $\rho_i$ can be defined as $\rho_i=\frac{C_i}{D_i}$. The semantic compression time $t_i^{\mathrm{comp}}$ is modeled as
\begin{equation}
t_i^{\mathrm{comp}}=
\frac{\alpha_i D_i \ln(1/\rho_i)}{f_i},
\label{eq:t_comp}
\end{equation}
where \(\alpha_i\) is the computational complexity coefficient for semantic extraction and \(f_i\) is the local CPU frequency. 

Then, agent $i$  semantically communicates with the BS, and the uplink transmission achievable rate \cite{10915662} is given by
\begin{equation}
R_i(p_i)=B\log_2\!\left(1+\frac{p_i |h_i|^2}{\sigma^2}\right),
\label{eq:R_i}
\end{equation}
where \(B\) is the channel bandwidth and \(p_i\) is the transmit power of agent \(i\).

The corresponding communication time $t_i^{\mathrm{comm}}$ is
\begin{equation}
t_i^{\mathrm{comm}}=
\frac{C_i}{R_i(p_i)}
=\frac{\rho_i D_i}{B\log_2\!\left(1+\frac{p_i |h_i|^2}{\sigma^2}\right)}.
\label{eq:t_comm}
\end{equation}

Therefore, the total data process time $t_i^{\mathrm{BS}}$ of agent $i$ is
\begin{equation}
\begin{split}
t_i^{\mathrm{BS}}&=t_i^{\mathrm{comp}}+t_i^{\mathrm{comm}} \\
&= \frac{\alpha_i D_i \ln(1/\rho_i)}{f_i}
 +\frac{\rho_i D_i}{B\log_2\!\left(1+\frac{p_i |h_i|^2}{\sigma^2}\right)}.
\end{split}
\label{eq:t_bs}
\end{equation}

In local computing mode, agent $i$ processes the raw data locally without semantic compression and wireless communicaton. Let \(\tau\) denote the number of CPU cycles required to process one bit of raw data. Therefore, the total data process time $t_i^{\mathrm{local}}$ of agent $i$ is $t_i^{\mathrm{local}} = \frac{\tau D_i}{f_i}$.

\subsection{Energy Consumption Model}

The total energy consumption of each agent \(i\) consists of information-processing energy and task-execution energy.

\subsubsection{Information-Processing Energy}In the BS-assisted collaborative mode, the information-processing energy contains semantic compression energy and uplink transmission energy. The semantic compression energy is $E_i^{\mathrm{comp}}=\kappa\, t_i^{\mathrm{comp}} f_i^3$, where \(\kappa\) is the effective switched capacitance coefficient \cite{10915662}. The uplink transmission energy is $E_i^{\mathrm{comm}}=p_i t_i^{\mathrm{comm}}$. Therefore, the information-processing energy is
\begin{equation}
\begin{split}
E_i^{\mathrm{BS}}&=E_i^{\mathrm{comp}}+E_i^{\mathrm{comm}} \\
&= \kappa \alpha_i D_i f_i^2 \ln(1/\rho_i) + \frac{p_i \rho_i D_i}{B\log_2\!\left(1+\frac{p_i |h_i|^2}{\sigma^2}\right)}.
\end{split}
\end{equation}

In the local computing mode, agent \(i\) only consumes local computation energy, so the information-processing energy is
\begin{equation}
E_i^{\mathrm{local}}
=\kappa\, t_i^{\mathrm{local}} (f_i)^3
=\kappa \tau D_i f_i^2.
\label{eq:E_loc}
\end{equation}

\subsubsection{Task-Execution Energy}
To mathematically characterize the task-execution energy by collaboration, we use a collaborative gain function to describe the energy consumption. The work in \cite{gunther1993simple} proposes a model for parallel processing performance in distributed systems, called universal scalability law (USL). The collaborative gain is formulated as
\begin{equation}
G(K) = (1-\beta) + \frac{\beta}{K} + \xi(K-1), \label{eq:G_beta}
\end{equation}
where $K =\sum_{i=1}^{N} x_i$ represents the total number of agents in the BS-assisted collaborative mode, and $K \in \{2,3,\dots,M\}$, $M=\sum_{i=1}^{N} a_i$. Collaboration coefficient $\beta \in (0, 1)$ denotes the proportion of the task execution energy that is compressible and dynamically optimized through collaboration, such as area scanning and distributed load carrying. $(1-\beta)$ denotes the lower bound of the energy required for basic mechanical operations regardless of collaboration scale, such as balancing and grabbing. $\xi \ge 0$ represents the contention coefficient, and $\xi(K-1)$ quantifies the additional energy degradation incurred by spatial interference, collision avoidance, and resource synchronization as the cluster size expands.

Therefore, the task execution energy model in dual-mode is
\begin{equation}
    E_{i}^{\text{move}} = 
    \begin{cases}
        Q \cdot G(K), & x_i = 1, \\
        Q, & x_i = 0,
    \end{cases}
\end{equation}
where $Q$ is the independent baseline task execution energy without any global collaboration.

\subsection{Problem Formulation}
Our objective is to jointly optimize the semantic compression ratio, transmit power, and operating mode to minimize the energy consumption of the MEAN in one collaboration round under a time constraint. The optimization problem is given by
\begin{subequations}
\begin{align}
\min_{\mathbf{x}, \boldsymbol{\rho}, \mathbf{p}} \quad &\sum_{i=1}^{N} x_i \left[E_i^{\mathrm{comm}} + E_i^{\mathrm{comp}}+Q\cdot G(K) \right], \nonumber \\ 
&+(1-x_i)(E_i^{\mathrm{local}} + Q), \tag{\theequation} \label{eq:all} \\
\text{s.t.}\quad
& x_i\, t_i^{\mathrm{BS}}
+(1-x_i)\, t_i^{\mathrm{local}}\le T_0,
\quad \forall i\in\mathcal{N},
\label{prob:latency}\\
& x_i \le a_i,\quad \forall i\in\mathcal{N},
\label{prob:channel}\\
& x_i\in\{0,1\},\quad \forall i\in\mathcal{N},
\label{prob:binary}\\
& K \in \{2,3,\dots,\sum_{i=1}^{N} a_i\},\quad \forall i\in\mathcal{N},
\label{prob:K}\\
& \rho_{\min} \le \rho_i\le 1,\quad \forall i\in\mathcal{N},
\label{prob:rho}\\
& 0< p_i\le P_{\max},\quad \forall i\in\mathcal{N},
\label{prob:power}
\end{align}
\end{subequations}
where $T_0$ is the maximum processing time requirement, $P_{\max}$ is the maximum transmit power of the agents, and $\rho_{\min}$ is the minimum compression ratio of the agents. The maximum processing time requirement for each agent is given in \eqref{prob:latency}. \eqref{prob:channel}-\eqref{prob:binary} guarantee that agent \(i\) can choose the BS-assisted collaborative mode only when $a_i=1$. The number of the agents in BS-assisted collaborative mode is given by \eqref{prob:K}. \eqref{prob:rho} represents the compression ratio constraint and \eqref{prob:power} represents the transmit power constraint.

\section{Algorithm Design}

The formulated optimization problem is a MINLP problem. Donate $\Psi(K) = K \cdot Q \cdot (1-G(K))$, and we reorganize the problem \eqref{eq:all} according to $x_i$, the problem can be rewritten as
\begin{equation}
\begin{split}
    \min_{\mathbf{x}, \boldsymbol{\rho}, \mathbf{p}} \quad & \underbrace{\sum_{i=1}^N (E_i^{\mathrm{local}} + Q)}_{\text{energy of agents all in local}} - \underbrace{\Psi(K)}_{\text{collaboration gain for task-execution}} \\
    &+ \sum_{i=1}^N x_i \underbrace{\left( E_i^{\mathrm{BS}} -E_i^{\mathrm{local}} \right)}_{\text{energy-saving potential}}. \label{eq:all1} \\
\end{split}
\end{equation}

In the proposed algorithm, the optimization problem is decoupled into two subproblem: semantic compression ratio and transmit power optimization and operating mode optimization. To optimize $(\mathbf{x}, \boldsymbol{\rho}, \mathbf{p})$ in problem \eqref{eq:all1}, we first optimize the $(\boldsymbol{\rho}, \mathbf{p})$ for agents in BS-assisted collaborative mode. Then, we optimize $(K,\mathbf{x})$ with fixed $(\boldsymbol{\rho}, \mathbf{p})$. 

\subsection{Semantic Compression Ratio and Transmit Power Optimization}

For any agent $i$ operating in the BS-assisted collaborative mode, minimize its energy-saving potential is equivalent to minimize the total energy. Thus, the problem \eqref{eq:all1} reduces to
\begin{subequations}
\begin{align}
\min_{\rho_i, p_i} \quad & E_i^{\text{BS}} = \kappa \alpha_i D_i f_i^2 \ln(1/\rho_i) + \frac{p_i \rho_i D_i}{B \log_2\left(1 + \frac{p_i |h_i|^2}{\sigma^2}\right)}, \tag{\theequation} \label{prob:1.obj} \\
\text{s.t.} \quad & t_i^{\text{BS}} \le T_0, \quad \forall i\in\mathcal{N}, \label{prob:1.1}\\
& \rho_{\min} \le \rho_i\le 1,\quad \forall i\in\mathcal{N},  \label{prob:1.2} \\
& 0 < p_i \le P_{\max} ,\quad \forall i\in\mathcal{N}.  \label{prob:1.3} 
\end{align}
\end{subequations}

At the optimal solution, the total processing time of agent $i$ must strictly satisfy the equality of the maximum latency constraint, so constraint \eqref{prob:1.1} can be formulated as $t_i^{\mathrm{comm}} = T_0 - t_i^{\mathrm{comp}} = T_0 + \frac{\alpha_i D_i}{f_i}\ln \rho_i$. 

To determine the feasible region $\Omega_i$ of $\rho_i$, considering $t_i^{\mathrm{comm}} > 0$, we have a physical lower bound $\rho_{\mathrm{inf}} = \max\left(\rho_{\min}, \exp\left(-\frac{T_0 f_i}{\alpha_i D_i}\right)\right)$. Furthermore, due to constraint \eqref{prob:1.3}, the transmission time must satisfy $t_i^{\mathrm{comm}} \ge \frac{\rho_i D_i}{B \log_2\left(1+\frac{P_{\max}|h_i|^2}{\sigma^2}\right)}$, which lacks a closed-form solution. Thus, we obtain the left-boundary threshold $\rho_{\mathrm{Pmax}}^L$ and $\rho_{\mathrm{Pmax}}^R$, and Consequently, the feasible region is definitively closed as $\Omega_i = [\max(\rho_{\min},\rho_{\mathrm{inf}}, \rho_{\mathrm{Pmax}}^L), \min(\rho_{\mathrm{Pmax}}^R,1)]$. 

Based on $\Omega_i$, if $\max(\rho_{\min},\rho_{\mathrm{inf}}, \rho_{\mathrm{Pmax}}^L) > \min(\rho_{\mathrm{Pmax}}^R,1)$, we have $\Omega_i = \emptyset$, implying that agent $i$ cannot satisfy the processing time requirement under $P_{\max}$, which enforces $a_i=0$, i.e., switching to the local computing mode. If $\rho_{\mathrm{Pmax}}^R > 1$, transmitting raw data inevitably causes timeout under severe channel degradation. Hence, forced semantic compression becomes the unique solution to maintain operations.

By substituting $t_i^{\mathrm{comm}}$, the subproblem is reduced to
\begin{equation}
\begin{split}
     \min_{\rho_i \in \Omega_i} \ E_i^{\text{BS}} = -\kappa \alpha_i D_i f_i^2 \ln \rho_i + \frac{\sigma^2 t_i^{\mathrm{comm}}}{|h_i|^2} \left( 2^{\frac{\rho_i D_i}{B t_i^{\mathrm{comm}}}} - 1 \right). \label{subproblem1}
\end{split}
\end{equation}

\begin{lemma}
    The subproblem \eqref{subproblem1} is strictly convex over $\Omega_i$. \label{lemma:convexity}
\end{lemma}

\begin{proof}
\emph{Let $P(v, t) = t(2^{v/t}-1)$ be the perspective function of $2^v-1$, which is jointly convex in $(v, t)$ for $t > 0$. The second term of $E_i^{\mathrm{BS}}$ can be rewritten as $\frac{\sigma^2}{|h_i|^2}P(v_i, t_i^{\mathrm{comm}})$, where $v_i = \frac{D_i}{B}\rho_i$ is an affine function, and $t_i^{\mathrm{comm}}$ is strictly concave. Since $\frac{\partial P}{\partial t} < 0$, applying the extended composition rules for convex optimization preserves the convexity. Combined with the first term $-\kappa \alpha_i D_i f_i^2 \ln \rho_i$, which is strictly convex, it rigorously holds that the objective function $E_i^{\mathrm{BS}}$ is strictly convex.} \hfill$\blacksquare$
\end{proof}

Based on Lemma \ref{lemma:convexity}, the subproblem admits a unique global minimum. By denoting $z_i = \frac{\rho_i D_i}{B t_i^{\mathrm{comm}}}$, the optimal unconstrained stationary point $\rho_{i}^{\mathrm{zero}}$ is strictly determined by
\begin{equation}
\begin{split}
&\frac{\sigma^2}{|h_i|^2} \left[ \frac{\alpha_i D_i}{f_i} \left(2^{z_i} - 1\right) + z_i 2^{z_i} \ln 2 \left( t_i^{\mathrm{comm}} - \frac{\alpha_i D_i}{f_i} \right) \right] \\
&= \kappa \alpha_i D_i f_i^2. \label{eq:stationary}
\end{split}
\end{equation}

Since a closed-form expression is intractable, we calculate this stationary point by employing the Bisection method to find the root of the derivative equation over the interval $(0, 1]$. Therefore, the exact optimal semantic compression ratio $\rho_i^*$ and the corresponding optimal transmit power $p_i^*$ are given by
\begin{align}
\rho_i^* &= \min \left( 1, \max \left( \rho_{\mathrm{inf}}, \rho_{\mathrm{Pmax}}, \rho_{i}^{\mathrm{zero}} \right) \right), \label{eq:rho_star}\\
p_i^* &= \frac{\sigma^2}{|h_i|^2} \left( 2^{\frac{\rho_i^* D_i}{B \cdot t_i^{\mathrm{comm}}}}-1\right). \label{eq:P_star}
\end{align}

\subsection{Operating Mode Optimization}

Given $(\boldsymbol{\rho}, \mathbf{p})$ calculated, the minimal information-processing energy $E_{i}^{\mathrm{BS*}}$ can be determined. With the optimal energy-saving potential defined as $\Delta E_i^{\mathrm{save*}} = E_i^{\mathrm{local}} - E_{i}^{\mathrm{BS*}}$, the problem \eqref{eq:all1} can be simplified as
\begin{subequations}\label{subprob_mode}
\begin{align}
    \min_{x_i} \quad & \sum_{i=1}^N (E_i^{\mathrm{local}} + Q) - \sum_{i=1}^N x_i \Delta E_i^{\mathrm{save*}} - \Psi(K), \tag{\theequation} \\
    \text{s.t.} \quad & x_i \le a_i, \quad \forall i \in \mathcal{N}, \\
    & x_i \in \{0,1\}, \quad \forall i \in \mathcal{N}, \\
    & K = \sum_{i=1}^N x_i, \quad \forall i \in \mathcal{N},\\
    & 2 \le K \le M, M = \sum_{i=1}^N a_i, \quad \forall i \in \mathcal{N}.
\end{align}
\end{subequations}

The problem \eqref{subprob_mode} lies in the non-linearity of $\Psi(K)$. However, $\Psi(K)$ depends on $K$, not the mode of the selected agents. By fixing $K$, $\Psi(K)$ degenerates into a constant. The subproblem then collapses into maximizing the sum of independent $\Delta E_{i}^{\mathrm{save*}}$. 

We define a permutation function $\Pi(j)$ that maps the sorted index $j \in \{1,2, \ldots, M\}$ to the original agent index $i \in \mathcal{N}$ for all feasible agents which $a_i = 1$. This function sorts these feasible agents in descending order based on $E_i^{\mathrm{save*}}$, ensuring that $\Delta E_{\Pi(1)}^{\mathrm{save*}} \ge \Delta E_{\Pi(2)}^{\mathrm{save*}} \ge \dots \ge \Delta E_{\Pi(M)}^{\mathrm{save*}}$.

We observe that the operating mode selection of each agent is independent and jointly constrained by the SNR and $K$. Therefore, the greedy selection of the top-$K$ $E_i^{\mathrm{save*}}$ avoids any integrality gap, providing the exact optimal discrete solution for a given $K$. By greedily enumerating the discrete finite space of valid $K \in \{2,3, \dots, M\}$ to find $K^*$, the global optimality of the entire problem is strictly preserved.

By prioritizing agents that offer the largest energy reductions, the optimal collaboration scale $K^*$ is determined by
\begin{equation}
    K^* = \arg\min_{K \in \{2,3, \ldots, M\}} \left\{ \Psi(K) - \sum_{j=1}^K \Delta E_{\Pi(j)}^{\mathrm{save*}} \right\}.
\end{equation}

Consequently, the final optimal mode selection solution $x_i^*$ for the MEAN is expressed as
\begin{equation}
    x_i^* = 
    \begin{cases} 
      1, & \text{if } i \in \{\Pi(1), \Pi(2), \ldots, \Pi(K^*)\}, \\
      0, & \text{otherwise}.
    \end{cases}
\end{equation}

\subsection{Algorithm Analysis}
For semantic compression ratio and transmit power optimization, determining the boundary threshold via Newton's method requires $\mathcal{O}(\log(1/\epsilon_1))$ operations, and the optimal resource variables converge in $\mathcal{O}(M \log(1/\epsilon_2))$ via line searches. For operating mode optimization, sorting the energy-saving potentials mandates $\mathcal{O}(M \log M)$, and the sequential sweep takes $\mathcal{O}(M)$, so the overall complexity is bounded by $\mathcal{O}(M^2 \log M + M^2 \log(1/\epsilon_2))$. Compared to the $\mathcal{O}(2^M)$ exponential complexity of standard MINLP solvers, the proposed algorithm has a lower complexity.

\begin{algorithm}[t]
\caption{Joint Resource Allocation and Mode Selection}
\label{algo:jcsra}
\begin{algorithmic}[1]
\REQUIRE System and agent parameters, initialize minimum energy $E_{\min} \leftarrow +\infty$, optimal scale $K^* \leftarrow 0$.
\STATE Compute SNR $\gamma_i$ and feasibility indicator $a_i$ to identify the $M$ feasible agents.
\FOR{each feasible agent $i$ ($a_i = 1$)}
    \STATE Compute optimal $\rho_i^*$, $p_i^*$, and $\Delta E_i^{\mathrm{save*}}$.
\ENDFOR
\STATE Sort feasible agents in descending order of their $\Delta E_i^{\mathrm{save*}}$.
\FOR{cluster size $K = 2$ to $M$ }
    \STATE Assign BS-assisted mode ($x_i=1$) to top-$K$ agents, compute $E_{\mathrm{sys}}^*(K)$, and update $E_{\min}$ and $K^*$ if a lower energy is found.
\ENDFOR
\RETURN Optimal $K^*$, $E_{\min}$, and $\{x_i^*, \rho_i^*, p_i^*\}_{i=1}^N$.
\end{algorithmic}
\end{algorithm}

\section{Simulation Results}\label{Sec:sr}
In the simulation, the simulation parameters used in the paper are summarized in table.~\ref{tab:sim_params}. For comparison, we evaluate the proposed algorithm (labeled as 'Proposed') against four baselines: a) a classical heuristic where the agent select the BS-assisted collaborative mode while satisfying the SNR threshold (labeled as 'SNR-Based'), b) all agents in local computing mode (labeled as 'Local Only'), c) all agents transmit raw data directly without semantic communication (labeled as 'No SemCom'), and d) all agents set the $p=0.5$ W and do not optimize it (labeled as 'Fixed Tx Power').

\begin{table}[t]
  \centering
  \caption{List of Simulation Parameters}                                                                                         \label{tab:sim_params}
  \begin{tabular}{lc | lc}                                                                                                \toprule
  \textbf{Parameter} & \textbf{Value} & \textbf{Parameter} & \textbf{Value} \\
  \midrule
  $B$ & $1$ MHz & $d_{\min}$ & $50$ m \\
  $P_{\max}$ & $1$ W & $d_{\max}$ & $1000$ m \\
  $\sigma^2$ & $4 \times 10^{-11}$ W & $N_{\text{mc}}$ & $1000$ \\
  $\gamma_{\text{th}}$ & $1$ & $\alpha$ & $10$ \\
  $T_0$ & $0.7$ s & $f$ & $1$ GHz \\
  $Q$ & $0.1$ J & $\kappa$ & $10^{-28}$ \\
  $\beta$ & $0.4$ & $\tau$ & $100$ \\
  $\xi$ & $0.008$ & $D$ & $10$ Mbits \\
  $N$ & $15$ & $\rho_{\min}$ & $0.1$ \\
  \bottomrule
  \end{tabular}
  \end{table}

Fig.~\ref{fig.N} illustrates the total energy consumption versus the number of agents. It is observed that the total energy consumption increases with the increase of the number of the agents, and the proposed algorithm achieves the best performance. This is due to the benefits of semantic communication and collaborative strategy. We show the total energy consumption versus the data size in Fig.~\ref{fig.D}. The total energy consumption increases with the data size for all algorithms since more data needs to be computed and more communication energy of the agents is used to satisfy the latency constraints. The total energy consumption versus the time constraint is presented in Fig.~\ref{fig.T0}. It is seen that the total energy consumption of the proposed algorithm and the SNR-Based algorithm decrease with the maximal latency. This is because relaxing latency constraints allow for longer computation and communication times, thereby reducing the energy consumption of semantic computation and communication transmission.

\begin{figure*}[t]
    \centering
    \vspace{-2em}
    \subfigure{
        \begin{minipage}{0.33\textwidth}
            \centering
            \includegraphics[width=1\textwidth]{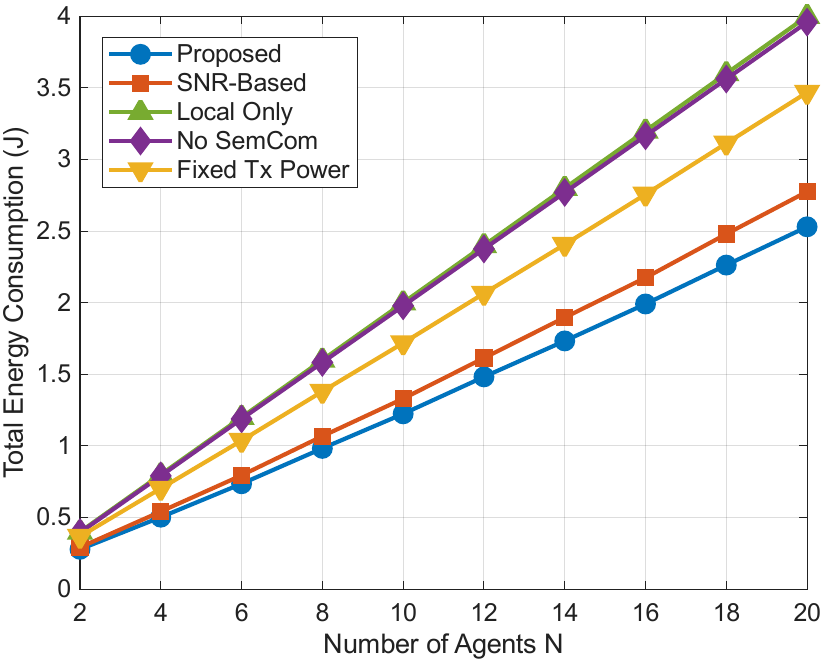}
            % \caption*{(a)}
            \label{fig.N}
    \end{minipage}}
    \hspace{-4mm}
    \subfigure{
        \begin{minipage}{0.33\textwidth}
            \centering
            \includegraphics[width=1\textwidth]{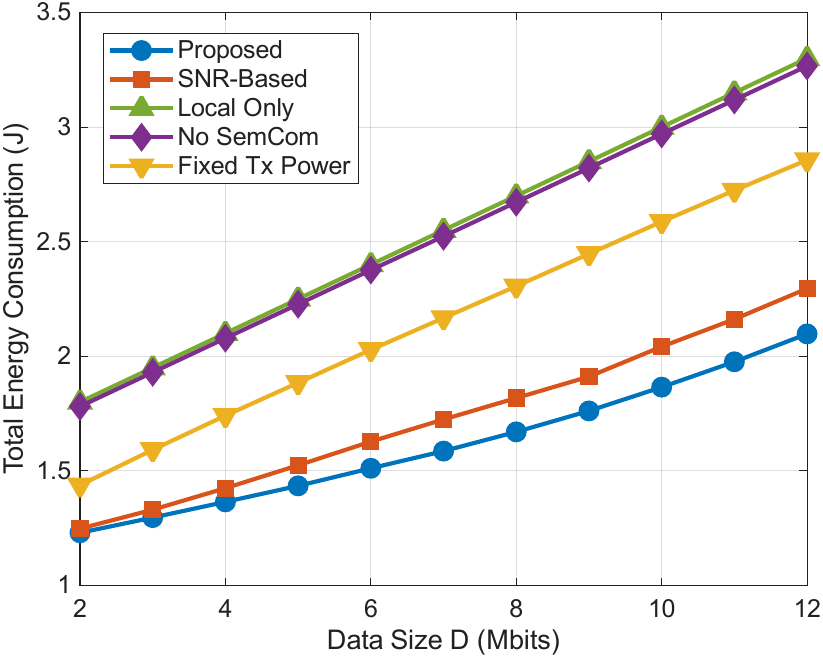}
            % \caption*{(b)}
            \label{fig.D}	
        \end{minipage}}
        \hspace{-4mm}
        \subfigure{
            \begin{minipage}{0.33\textwidth}
                \centering
                \includegraphics[width=1\textwidth]{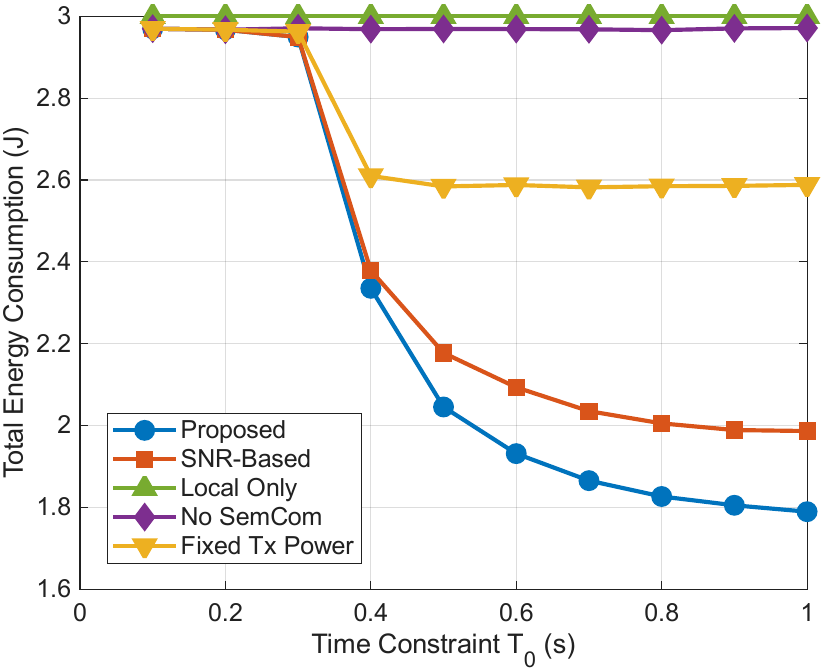}
                % \caption*{(c)}
                \label{fig.T0}	
        \end{minipage}}
    \vspace{-1.4em}
    % \captionsetup{justification=raggedright,singlelinecheck=false}
    \caption{Total energy consumption versus: (a) Number of agents $N$, (b) Data size $D$, (c) Time constraint $T_0$.}
    \label{fig.sim} 
    \vspace{-1em}
\end{figure*}

\section{Conclusion}\label{Sec:c}
In this letter, we have investigated the computation and communication resource allocation and mode selection problem for MEAN over wireless communication. The semantic compression ratio, transmit power, and operating mode are jointly optimized to maximize the system energy efficiency. To solve this problem, we have proposed two-stage numeration algorithm with low complexity. Numerical results have shown that the proposed algorithm scheme outperforms benchmark schemes in terms of energy efficiency, especially for large number of agents and large data size.

\bibliographystyle{IEEEtran}
\bibliography{ref}

\end{document}